\documentclass[11pt,a4paper]{article}

\usepackage{jheppub}
\usepackage{psfrag}
\usepackage{slashed}
\usepackage{cancel}
\usepackage{lscape}

\def\la{\langle}
\def\ra{\rangle}

\definecolor{mygreen}{rgb}{0,0.7,0}

\def\cN{\mathcal{N}}

\def\A#1#2{\la#1#2\ra}

\def\AB#1#2#3{\la#1|#2|#3]}

\def\eps{\epsilon}

\def\qb{{\bar{q}}}
\def\sb{{\bar{s}}}

\newcommand{\define}{\equiv}

\definecolor{myblue}{rgb}{0,0.6,0.6}

\usepackage{amsthm}
\newtheorem{lemma}{Lemma}

\preprint{}

\title{An Integrand Reconstruction Method for Three-Loop Amplitudes}

\author[a]{Simon Badger}
\author[a]{Hjalte Frellesvig}
\author[a]{Yang Zhang}
\affiliation[a]{
Niels Bohr International Academy and Discovery Center, The Niels Bohr Institute,\\%
University of Copenhagen, Blegdamsvej 17, DK-2100 Copenhagen, Denmark}
\emailAdd{badger@nbi.dk,hjf@nbi.dk,zhang@nbi.dk}

\abstract{
We consider the maximal cut of a three-loop four point function with massless kinematics.
By applying Gr\"obner bases and primary decomposition we develop a method which extracts
all ten propagator master integral coefficients for an arbitrary triple-box configuration via generalized unitarity cuts.
As an example we present analytic results for the three loop triple-box contribution to gluon-gluon scattering
in Yang-Mills with adjoint fermions and scalars in terms of three master integrals.
}

\keywords{}

\begin{document}
\maketitle
\flushbottom

\section{Introduction}

Precise predictions for cross sections in collider experiments require knowledge of scattering
amplitudes to a high loop order. Computations of this type, particularly in QCD, are still
unfortunately out of reach using known techniques. Unitarity cuts have shown to be an enormously
powerful tool for super-symmetric gauge and gravity theories with full amplitudes now computed up to
four loops \cite{Bern:2009kd,Bern:2010tq}. At three-loops, steps beyond maximally super-symmetric amplitudes have
also been considered in the recent computation of the UV behaviour of graviton-graviton scattering in
$\mathcal{N}=4$ supergravity \cite{Bern:2012cd}. For a recent review of state of the techniques see \cite{Carrasco:2011hw}
and references therein.

The success of unitarity \cite{Bern:1994zx,Bern:1994cg} and generalised unitarity
\cite{Britto:2004nc} in automating multi-loop one-loop corrections has sparked some recent interest
in exploring the application of integrand reduction (OPP-like \cite{Ossola:2006us}) methods at
two-loops
\cite{Mastrolia:2011pr,Kosower:2011ty,Badger:2012dp,Larsen:2012sx,CaronHuot:2012ab,Kleiss:2012yv}.
Very recently the possibility of using computational algebraic geometry to overcome the traditional
bottlenecks in amplitude computations has started to be explored
\cite{Zhang:2012ce,Mastrolia:2012an}.

The traditional approach to a unitarity computation of a loop amplitude relies on the knowledge of
a basis of known integral functions. At two loops it has been possible to derive sets of
integration-by-parts identities \cite{Chetyrkin:1981qh,Gehrmann:1999as} that reduce amplitudes to a unitarity
compatible basis \cite{Gluza:2010ws,Schabinger:2011dz}. Using an integrand parametrisation
constrained by the ideal generated from the propagators with help from Gram matrix identities we can also reduce a unitarity compatible
form of the amplitude to master integrals using the well known Laporta algorithm \cite{Laporta:2001dd} which
is implemented in a number of public codes \cite{Anastasiou:2004vj,Smirnov:2008iw,Studerus:2009ye,vonManteuffel:2012yz}.

In this paper we consider the extension of the new integrand reduction techniques to three-loop
amplitudes. We derive a complete reduction to ten-propagator master integrals (MI) for the maximal cut of
the three loop planar triple box with four massless external legs. As a first application
of the technique we compute the ten-propagator MI coefficients for gluon scattering in Yang-Mills theory with
adjoint fermions and scalars. The result applies in both $\mathcal{N}=4,2,1$ and $0$ super-symmetric
theories.

\section{Reduction of the Massless Triple Box \label{sec:triple}}

\begin{figure}[b]
  \centering
  \psfrag{1}{\small$p_1$}
  \psfrag{2}{\small$p_2$}
  \psfrag{3}{\small$p_3$}
  \psfrag{4}{\small$p_4$}
  \psfrag{l1}{\small$l_1$}
  \psfrag{l2}{\small$l_2$}
  \psfrag{l3}{\small$l_3$}
  \psfrag{l4}{\small$l_4$}
  \psfrag{l5}{\small$l_5$}
  \psfrag{l6}{\small$l_6$}
  \psfrag{l7}{\small$l_7$}
  \psfrag{l8}{\small$l_8$}
  \psfrag{l9}{\small$l_9$}
  \psfrag{l10}{\small$l_{10}$}
  \includegraphics[width=0.8\textwidth]{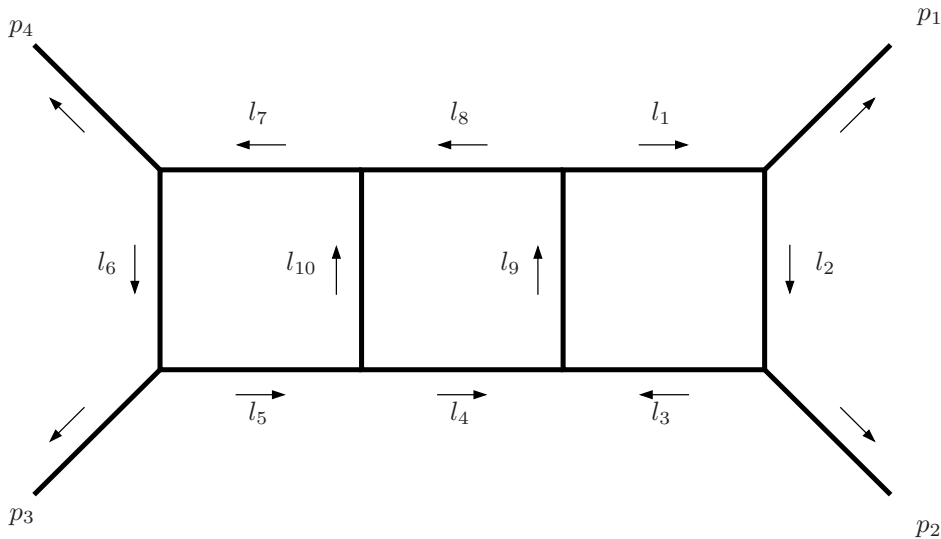}
  \caption{The momentum flow and propagator definitions for the three-loop planar triple box.}
  \label{fig:boxboxbox}
\end{figure}

In this paper we will study the three-loop planar triple box topology, shown in figure \ref{fig:boxboxbox}, defined by:

\begin{align}
  I_{10}[\tilde{N}_D(\eps,k_i,p_j)] =
  \int \! \frac{d^D k_1}{(2\pi)^D}
  \int \! \frac{d^D k_2}{(2\pi)^D}
  \int \! \frac{d^D k_3}{(2\pi)^D}
  \frac{ \tilde{N}_D(\eps,k_i,p_j) }
  {\prod_{i=1}^{10} l_i^2 },
  \label{eq:tripleboxdef}
\end{align}

where the ten propagators $\{l_i\}$ are given by:
\begin{align}
  l_1 &= k_1, & l_2 &= k_1 - p_1, & l_3 &= k_1 - p_1 - p_2, & l_4 &= k_3 + p_1 + p_2, \nonumber \\
  l_5 &= k_2 + p_1 + p_2, & l_6 &= k_2 - p_4, & l_7 &= k_2, & l_8 &= k_3, \nonumber \\
  l_9 &= k_1 + k_3, & l_{10} &= k_2 - k_3. &&
  \label{eq:lmomdefs}
\end{align}

The external momenta, $\{p_i\}$, and internal momenta, $\{l_i\}$, are considered massless.
Though the integral needs to be dimensionally regulated in $D=4-2\eps$ dimensions we will
only consider generalized unitarity cuts in four dimensions and therefore reconstruct only
the leading term of the numerator function, $\tilde{N}_D(\eps,k_i,p_j) = \tilde{N}(k_i,p_j) +
\mathcal{O}(\eps)$.

We also consider the external momenta to be outgoing with the Mandelstam variables defined by:
\begin{align}
  s &= (p_1+p_2)^2, & t &= (p_1+p_4)^2, & u &= -s-t.
  \label{eq:mandelstam}
\end{align}

Following a similar approach taken at two loops in our previous work \cite{Badger:2012dp}, we will
proceed in three steps: firstly the on-shell constraints for the deca-cut will be solved. Secondly,
we invert a linear system to map the polynomial of the deca-cut defined by the on-shell constraints
onto the general integrand basis. Finally the integrand is reduced to master integrals(MIs) using
additional integration-by-parts(IBP) relations.

\subsection{Solving The On-shell Constraints}

The parametrisation of the loop-momenta using two component Weyl spinors has been chosen as follows:
\begin{align}
l_2 &= x_1 p_1 + x_2 p_2 + x_3 \frac{\A{2}{3}}{\A{1}{3}} \frac{\langle p_1 | \gamma^{\mu} | p_2
]}{2} + x_4 \frac{\A{1}{3}}{\A{2}{3}} \frac{\langle p_2 | \gamma^{\mu} | p_1 ]}{2}, \nonumber \\
-l_6 &= y_1 p_3 + y_2 p_4 + y_3 \frac{\A{4}{1}}{\A{3}{1}} \frac{\langle p_3 | \gamma^{\mu} | p_4
]}{2} + y_4 \frac{\A{3}{1}}{\A{4}{1}} \frac{\langle p_4 | \gamma^{\mu} | p_3 ]}{2}, \nonumber \\
-l_4 &= z_1 p_2 + z_2 p_3 + z_3 \frac{\A{3}{4}}{\A{2}{4}} \frac{\langle p_2 | \gamma^{\mu} | p_3
]}{2} + z_4 \frac{\A{2}{4}}{\A{3}{4}} \frac{\langle p_3 | \gamma^{\mu} | p_2 ]}{2}.
\label{loopbasis}
\end{align}
The ten cut-constraints $l_i^2 = 0$ form a set of ideals. Following the recent method proposed by
Zhang \cite{Zhang:2012ce}, this system can be reduced using primary
decomposition to find 14 independent solutions. These solutions come in complex conjugate\footnote{Technically explicit complex conjugation is only valid for real external momenta.
However, the solutions $1'-7'$ are also valid for complex external momenta.}
pairs and we label them $1-7$ and $1'-7'$. We find that all 14 solutions have the same
dimension and therefore can be parametrised with two variables which we call $\tau_1$ and $\tau_2$.
For the solutions $1-7$, explicit forms for the coefficients in (\ref{loopbasis}) can be written as in tables \ref{xytable} and \ref{ztable}.

\begin{table}[h]
\centering
\(\begin{array}{|c|c|c|c|c|}
\hline
 & x_1 & x_2 & x_3 & x_4 \\ \hline
1 & 0 & 0 & 1-\frac{s}{t} \frac{1+\tau_1}{\tau_2} & 0 \\
2 & 0 & 0 & -\frac{u}{t}(1+\tau_1) & 0 \\
3 & 0 & 0 & \tau_2 & 0 \\
4 & 0 & 0 & -\frac{u}{t}(1+\tau_1) & 0 \\
5 & 0 & 0 & \tau_2 & 0 \\
6 & 0 & 0 & 1 & 0 \\
7 & 0 & 0 & \tau_1 & 0 \\ \hline
\end{array} \;\;\;\;\; \)
\(\begin{array}{|c|c|c|c|c|}
\hline
 & y_1 & y_2 & y_3 & y_4 \\ \hline
1 & 0 & 0 & 1+\tau_2 & 0 \\
2 & 0 & 0 & 0 & -1+\frac{s}{u}\frac{1}{\tau_1} \\
3 & 0 & 0 & 0 & -1-\tau_1 \\
4 & 0 & 0 & 0 & \tau_2 \\
5 & 0 & 0 & 1 & 0 \\
6 & 0 & 0 & \tau_2 & 0 \\
7 & 0 & 0 & \tau_2 & 0 \\ \hline
\end{array}\)
\caption{Values of the coefficients in $l_2$ and $-l_6$ in (\ref{loopbasis}).}
\label{xytable}
\end{table}

\begin{table}[h]
\centering
\(\begin{array}{|c|c|c|c|c|}
\hline
 & z_1 & z_2 & z_3 & z_4 \\ \hline
1 & \frac{s}{t} \left( 1 + \frac{1}{\tau_1} \right) \! \left( 1 + \frac{1}{\tau_2} \right) & -\frac{s}{t} \left( 1 + \frac{1}{\tau_1} \right) + \frac{\tau_2}{\tau_1} & -\frac{u}{t} \left( 1 + \frac{1}{\tau_1} \right) & Z_{41} \\
2 & \frac{s}{t} \big( 1+\tau_2 \big) \! \left( 1 - \frac{u}{s} \tau_1 \right) & \frac{s}{t}\left( 1 + \frac{1}{\tau_1} \right) \tau_2 & \left( \frac{u}{t} - \frac{s}{t} \frac{1}{\tau_1} \right) \tau_2 & -\frac{s}{t}(1+\tau_1)(1+\tau_2) \\
3 & \frac{s}{t} \left( 1 + \frac{1}{\tau_1} \right) & 0 & 0 & \frac{s}{t} \left( \frac{s}{u} \frac{1}{\tau_1} - 1 \right) \\
4 & 0 & -\frac{s}{t} \left( 1 + \frac{1}{\tau_1} \right) & \frac{s}{t} \frac{1}{\tau_1} - \frac{u}{t} & 0 \\
5 & \tau_1 & 0 & 0 & \frac{s}{u} ( 1 + \tau_1 ) \\
6 & 0 & \tau_1 & 0 & \frac{s}{u} ( 1 - \tau_1 )  \\
7 & 0 & 0 & 0 & \frac{s}{u} \\ \hline
\end{array}\)
\caption{Values of the coefficients in $-l_4$ in (\ref{loopbasis}). \(Z_{41} \define \frac{s^2}{u t} \left( 1 + \frac{1}{\tau_1} \right) \! \left( 1 + \frac{1}{\tau_2}
\right) - \frac{s}{u} \frac{1 + \tau_2}{\tau_1}\).}
\label{ztable}
\end{table}

It is worth mentioning that the parametrisation is Laurent-polynomial, so that no terms of the form
\(\frac{1}{1+\tau}\) appear. We will therefore be able to fit the integrand using an efficient
discrete Fourier projection just as used succesfully at one-loop
\cite{Berger:2008sj,Mastrolia:2008jb}.

The conjugate solutions can be easily constructed from the above expressions using,
\begin{align}
x_1^{s'}  &= 0 & x_2^{s'}  &= 0 &
  x_3^{s'}  &= \tfrac{u}{t} x_4^s  & x_4^{s'}  &= \tfrac{t}{u} x_3^s  \nonumber \\
y_1^{s'}  &= 0 & y_2^{s'}  &= 0 &
  y_3^{s'}  &= \tfrac{u}{t} y_4^s  & y_4^{s'}  &= \tfrac{t}{u} y_3^s  \nonumber \\
z_1^{s'}  &= z_1^s  & z_2^{s'}  &= z_2^s  &
  z_3^{s'}  &= \tfrac{u}{s} z_4^s  & z_4^{s'}  &= \tfrac{s}{u} z_3^s.
\label{xyzprime}
\end{align}
Note that we have suppressed the functional dependence of the coefficients, e.g. $x_3^{s}\equiv
x_3^{s}(\tau_1,\tau_2)$.

\subsection{Fitting the Integrand Basis \label{sec:basis}}

The space of loop momenta is spanned by three external momenta, which we choose to be
$\{p_1,p_2,p_4\}$, and an additional orthogonal direction $\omega$ which is defined by,
\begin{align}
  \omega &\define \frac{1}{2s} \Big( \AB{2}{3}{1} \AB{1}{\gamma^{\mu}}{2} - \AB{1}{3}{2}
  \AB{2}{\gamma^{\mu}}{1} \Big).
\end{align}

The integrand numerator $N$, can be parametrised in terms of seven irreducible
scalar products (ISPs) as follows:
\begin{align}
  N = \sum_{\alpha_i} &c_{\alpha_1\dots\alpha_7}
  (k_1 \cdot p_4)^{\alpha_1}
  (k_2 \cdot p_1)^{\alpha_2}
  (k_3 \cdot p_4)^{\alpha_3}
  (k_3 \cdot p_1)^{\alpha_4}\nonumber\\&\times
  (k_1 \cdot \omega)^{\alpha_5}
  (k_2 \cdot \omega)^{\alpha_6}
  (k_3 \cdot \omega)^{\alpha_7}.
  \label{eq:integrandbasis}
\end{align}

We can use renormalizability conditions to restrict the maximum powers of $\alpha_i$ and complete
the reduction by using polynomial division with respect to a Gr\"obner basis constructed from the
propagators constraints. This procedure is implemented in the Mathematica package
{\tt BasisDet} and we refer to \cite{Zhang:2012ce} for full details of the method.
The resultant integrand naturally splits into 199 spurious (S) and 199 non-spurious (NS) terms of the
form:
\begin{align}
  N^{\rm NS} &= \sum_{\alpha_i}
  (k_1 \cdot p_4)^{\alpha_1}
  (k_2 \cdot p_1)^{\alpha_2}
  (k_3 \cdot p_4)^{\alpha_3}
  (k_3 \cdot p_1)^{\alpha_4}\nonumber\\&\times
  \left( c^{\rm NS}_{\alpha_1\dots\alpha_4 0}+c^{\rm NS}_{\alpha_1\dots\alpha_4 1}(k_1 \cdot \omega)(k_2 \cdot \omega) \right),
  \label{eq:integrandbasisNS} \\
  N^{\rm S} &= \sum_{\alpha_i}
  (k_1 \cdot p_4)^{\alpha_1}
  (k_2 \cdot p_1)^{\alpha_2}
  (k_3 \cdot p_4)^{\alpha_3}
  (k_3 \cdot p_1)^{\alpha_4}\nonumber\\&\times
  \left(
   c^{\rm S}_{\alpha_1\dots\alpha_4 0} (k_1 \cdot \omega)
 + c^{\rm S}_{\alpha_1\dots\alpha_4 1} (k_2 \cdot \omega)
 + c^{\rm S}_{\alpha_1\dots\alpha_4 2} (k_3 \cdot \omega)
  \right). \label{eq:integrandbasisS}
\end{align}
The spurious terms will vanish after integration since they are linear in $k_i \cdot \omega$.

Inserting the expressions for the loop-momenta given in \eqref{loopbasis} into the integrand
$N = N^{NS}+N^{S}$, defines 14 Laurent-polynomials of $\tau_1$ and $\tau_2$,
\begin{align}
  N |_{s} &= \sum_{ij} d_{sij} \tau_1^i \tau_2^j,
  \label{dexpansion}
\end{align}
where $s$ denotes one of the on-shell solutions either $1-7$ or $1'-7'$. In total we find $622$
terms in these expansions which we collect into a vector, $\bf d$.
After collecting the 398 coefficients of \eqref{eq:integrandbasis} into a vector $\bf c$
we can define a Matrix $M$ such that,
\begin{align}
  \mathbf{d} = M \mathbf{c}.
\end{align}
After inverting this matrix we can re-construct an arbitrary integrand by using the Laurent expansion
of the product of eight tree-level amplitudes to extract the values of the coefficients, $d_{sij}$.

Though in principle we could invert the matrix in its full $622\times398$ form this task is
quite complicated in practice. However, it is straightforward to separate the problem into two pieces
corresponding the spurious and non-spurious terms. We can show that,
\begin{align}
  k^s(\tau_1,\tau_2) \cdot p &= k^{s'}(\tau_1,\tau_2) \cdot p, &
  k^s(\tau_1,\tau_2) \cdot \omega &= -k^{s'}(\tau_1,\tau_2) \cdot \omega.
  \label{plusminusrelation}
\end{align}
Therefore when we combine the paired solutions we find,
\begin{align}
  N |_{s} + N |_{s'} &= N^{\rm NS} |_{s} + N^{\rm NS} |_{s'} = 2N^{\rm NS}
  |_{s}, \\
  N |_{s} - N |_{s'} &= N^{\rm S} |_{s} - N^{\rm S} |_{s'} = 2N^{\rm S} |_{s}.
  \label{eq:combine}
\end{align}
After defining two new vectors, ${\bf d}_\pm$, where
\begin{align}
  d_{\pm}^s = \frac{d^s \pm d^{s'}}{2},
\end{align}
we find smaller matrices $311\times199$, $M_\pm$, satisfying
\begin{align}
  \mathbf{d}_+ &= M_{+} \mathbf{c}^{\rm NS}, &
  \mathbf{d}_- &= M_{-} \mathbf{c}^{\rm S}.
  \label{eq:S+NSmatrix}
\end{align}
In this form it was possible to invert the matrices analytically using standard computer algebra
packages.

\subsection{Alternative Branch-by-branch Fitting}

In this section we describe a technique that allows each branch of the on-shell solutions to be
considered as a separate linear system, a natural extension of the strategy taken at the end of the
previous section when we separated spurious and non-spurious terms.  This will be achieved by first
partially fitting the integrand on the $14$ solutions to get $14$ polynomials. Then the $14$
polynomials are combined together to get the full integrand, by a Gr\"obner basis method. Some of
the more mathematical details are described in the Appendix \ref{app:branch}.

On each of the $14$ solutions the integrand numerator $N$ can be reduced further to a polynomial
$N_i$ with much fewer monomials, by the multivariate division towards the Gr\"obner basis of the
corresponding branch. We note that $N_i\equiv N_i(\{\rm ISPs\})$ is distinct from $N|_i\equiv
N|_i(\tau_1,\tau_2)$. The number of monomials of each $N_i$ is listed in table \ref{branch-by-branch}.

\begin{table}[h]
\centering
\begin{tabular}{|c|c||c|c|}
\hline
 & number of terms & &number of terms\\
\hline
1 & 75 & 1$'$ & 75\\
2 & 59 & 2$'$ & 59\\
3 & 42 & 3$'$ & 42\\
4 & 42 & 4$'$ & 42\\
5 & 22 & 5$'$ & 22\\
6 & 22 & 6$'$ & 22\\
7 & 25 & 7$'$ & 25\\
\hline
\end{tabular}
\caption{Number of terms for reduced integrand in each solution.}
\label{branch-by-branch}
\end{table}

For example, divided by the Gr\"obner basis of the second branch, the
integrand is reduced to a linear combination of $59$ terms,
\begin{gather}
  \label{eq:5}
  \big\{x^4 y^4,x^4 y^3,x^3 y^4,x^4 z^2,x^4 y^2,x^3 z^3,x^3 y^3,x^2 z^4,x^2 y^4,x^4
   z,x^4 y,x^3 z^2,\nonumber  \\  x^3 y^2,x^2 z^3,
x^2 y^3,x z^4,x y^4,x^4,x^3 z,x^3 y,x^2 z^2,x^2
   y^2, x z^3,x y^3,z^4,y^4,x^3,\nonumber  \\  x^2 z,x^2 y,x z^2,x y^2,z^3,y^3,x^2,x z,x
   y,z^2,y^2,x,z,y,1\big\},
\end{gather}
where $x=k_1 \cdot p_4$, $y=k_2 \cdot p_1$, $z=k_3\cdot p_4$.
The coefficients are determined by polynomial fitting
techniques, at the second solution.
\begin{equation}
  \label{eq:6}
  \boldsymbol d^{(2)}=M^{(2)} \boldsymbol c^{(2)},
\end{equation}
where $\boldsymbol c^{(2)}$ is the list of the $59$ coefficients and $\boldsymbol d^{(2)}$
contains the Laurent expansion coefficients at the second
solution. $M^{(2)}$ is a $62\times 59$ matrix, which is much smaller
than $M$, so it is easy to compute $\boldsymbol c^{(2)}$ and the reduced integrand
$N_2$.

Once all $14$ $N_i$'s are obtained, we ``merge'' them together to get the integrand $N$. The step is
equivalent to solving congruence equations in a polynomial ring. It is automatically done by a
Macaulay2 program \cite{M2}, based on a computation using Gr\"obner basis and ideal intersection. A more detailed
mathematical description of the procedure is outlined in Appendix \ref{app:branch}.

\section{Results for gluon-gluon scattering in Yang Mills theories \label{sec:4ptMIcoeffs}}

The starting expression for the computation is the product of tree-level amplitudes:
\begin{align}
  N|_s &= \sum_{f_i} \sum_{h_i} F_{f_1\dots f_5}(n_f,n_s)
  \; A^{(0)}_3 \! (-l_{1,f_1}^{-h_1}, p_1^{\lambda_1}, l_{2,f_1}^{h_2})
  \; A^{(0)}_3 \! (-l_{2,f_1}^{-h_2}, p_2^{\lambda_2}, l_{3,f_1}^{h_3}) \nonumber \\& \;\;\; \times
  A^{(0)}_3 \! (-l_{3,f_1}^{-h_3}, -l_{4,f_2}^{-h_4}, l_{9,f_4}^{h_9})
  \; A^{(0)}_3 \! (l_{4,f_2}^{h_4}, -l_{5,f_3}^{-h_5}, l_{10,f_5}^{h_{10}})
  \; A^{(0)}_3 \! (l_{5,f_3}^{h_5}, p_3^{\lambda_3}, -l_{6,f_3}^{-h_6}) \nonumber \\& \;\;\; \times
  A^{(0)}_3 \! (l_{6,f_3}^{h_6}, p_4^{\lambda_4}, -l_{7,f_3}^{-h_7})
  \; A^{(0)}_3 \! (l_{7,f_3}^{h_7}, -l_{8,f_2}^{-h_8}, -l_{10,f_5}^{-h_{10}})
  \; A^{(0)}_3 \! (l_{8,f_2}^{h_8}, l_{1,f_1}^{h_1}, -l_{9,f_4}^{-h_9}) ,
  \label{eq:treeprod}
\end{align}
where $l_i$ and $p_j$ are defined in \eqref{eq:lmomdefs}. $\{h_i\}$ are the internal helicity
states and $\{\lambda_i\}$ are the helicities of the external gluons. There are 34 distinct
configurations of the internal flavours $\{f_i\}$ each associated with a number of fermion flavours,
$n_f$ and complex scalar flavours, $n_s$. The explicit expressions for the flavour coefficients
$F_{f_1\dots f_5}(n_f,n_s)$ are given in Appendix \ref{app:flav}.

Complete expressions for all integrand coefficients have been found as functions of $n_f$ and $n_s$.
Though each term is relatively compact, the full set of 398 coefficients makes the expression rather
lengthy so we only present the result after further reduction to MIs. The full expressions for the
integrands are available from the authors on request. We note however that all coefficients in the
Laurent expansion \eqref{dexpansion} with $|i|$ or $|j| > 4$ vanish for all $n_f$ and $n_s$
demonstrating the general renormalization conditions for this case are simpler than assumed when
constructing the integrand in section \ref{sec:basis}.

IBP relations generated with {\tt Reduze2} \cite{Studerus:2009ye,vonManteuffel:2012yz} reduced the
199 non-vanishing integrals in the triple box integrand onto 3 Master Integrals, $I_{10}[1]$,
$I_{10}[(k_1+p_4)^2]$ and $I_{10}[(k_3-p_4)^2]$. The analytic expression for the scalar integral $I_{10}[1]$ has
been known already for some time \cite{Smirnov:2003vi}. We should be clear that to obtain the reduction of a full amplitude to
all master integrals the procedure must be carried out in two steps. Firstly the full integrand, without applying IBPs, must be kept
as a subtraction term for each lower propagator topology. Secondly IBP relations, again including all lower propagator MIs, should be
applied to the complete amplitude.

In order to map the basis of 7 ISPs of \eqref{eq:integrandbasisNS} onto the 15
propagator\footnote{For the purposes of the IBPs we consider both positive and negative powers of
the propagators in the topology hence the 15 propagators includes both the 10 denominators of eq.
\eqref{eq:tripleboxdef} and the 5 non-spurious ISPs.} integral topology obeying IBP relations we must re-write the non-spurious term $(k_1 \! \cdot \! \omega) \; (k_2 \! \cdot \! \omega)$
in terms of the reducible quantity $k_1 \! \cdot \! k_2$. This is achieved straightforwardly using the Gram determinant identity
generated by:
\begin{align}
  \text{det}
  \begin{pmatrix}
    1 & 2 & 4 & \omega & k_1\\
    1 & 2 & 4 & \omega & k_2
  \end{pmatrix}
  = 0,
\end{align}
which is re-written as:
\begin{equation}
  (k_1\cdot\omega) (k_2\cdot\omega) = -\frac{t^2}{4}+\frac{t}{2} \big( (k_1\cdot4)+(k_2\cdot1) \big)
  +\frac{tu}{s} (k_1\cdot k_2) + \frac{s+2t}{s} (k_1\cdot4) (k_2\cdot1).
\end{equation}
We note that this does not change the number of non-spurious integrals in the integrand basis.

We write the 4-point 3-loop primitive amplitude for this ladder topology as:
\begin{align}
  A^{(3)}_4(1^{\lambda_1},2^{\lambda_2},3^{\lambda_3},4^{\lambda_4}) &=
    C_1(1^{\lambda_1},2^{\lambda_2},3^{\lambda_3},4^{\lambda_4}) I_{10}[1] \nonumber\\&
  + C_2(1^{\lambda_1},2^{\lambda_2},3^{\lambda_3},4^{\lambda_4}) I_{10}[(k_1+p_4)^2] \nonumber\\&
  + C_3(1^{\lambda_1},2^{\lambda_2},3^{\lambda_3},4^{\lambda_4}) I_{10}[(k_3-p_4)^2]
  + \ldots
  \label{eq:3lamp_ladder}
\end{align}
where the ellipses cover terms with less than ten propagators. We then define dimensionless coefficients $\hat{C}_k$ by,
\begin{equation}
  C_k(1^{\lambda_1},2^{\lambda_2},3^{\lambda_3},4^{\lambda_4}) =
  s^3t A^{(0)}_4(1^{\lambda_1},2^{\lambda_2},3^{\lambda_3},4^{\lambda_4}) \hat{C}^{\lambda_1\lambda_2\lambda_3\lambda_4}_k(s,t),
  \label{eq:coeffnorm}
\end{equation}
where $A_4^{(0)}$ are the tree-level helicity amplitudes (expressions for the tree-level amplitudes are
collected in Appendix \ref{app:flav} for convenience).
After following the procedure of computing the Laurent series of \eqref{eq:treeprod},
reconstructing the integrand and applying the IBPs we obtain the following results for the
three independent helicity configurations. The $--++$ turns out to be trivial,
\begin{align}
  \hat{C}^{--++}_1(s,t) &= -1, \\
  \hat{C}^{--++}_2(s,t) &= 0, \\
  \hat{C}^{--++}_3(s,t) &= 0.
  \label{eq:3lmmpp}
\end{align}
The $-++-$ configuration is,
\begin{align}
  \hat{C}^{-++-}_1(s,t) &= \nonumber\\&
  - 1
  - (4-n_f)(3-n_s) \frac{s(t+2s)}{2t^2}
  + (4-n_f) \frac{s(t+4s)}{2t^2} \nonumber\\&
  - (1+n_s-n_f) \frac{s}{t^3}\left( 2t^2 + 11st + 10s^2 \right),
  \\
  \hat{C}^{-++-}_2(s,t) &= \frac{2}{t}\left(1-\hat{C}^{-++-}_1 \right),\\
  \hat{C}^{-++-}_3(s,t) &= 0,
  \label{eq:3lmppm}
\end{align}
and finally the alternating $-+-+$,
\begin{align}
  \hat{C}^{-+-+}_1(s,t) &= \nonumber\\&
  - 1
  + (4-n_f) \frac{st}{u^2}
  - 2(1+n_s-n_f) \frac{s^2 t^2}{u^4} \nonumber\\&
  + \big( 2(1-2n_s)+n_f \big)(4-n_f) \frac{s^2t(2t-s)}{4u^4} \nonumber\\&
  - \big( n_f(3-n_s)^2-2(4-n_f)^2 \big) \frac{st(t^2-4st+s^2)}{8u^4},
  \\
  \hat{C}^{-+-+}_2(s,t) &= \nonumber\\&
  - (4-n_f) \frac{s}{u^2}
  + 2(1+n_s-n_f) \frac{s^2 t}{u^4} \nonumber\\&
  - \big( 2(1-2n_s)+n_f \big)(4-n_f) \frac{s^2(2t-s)}{u^4} \nonumber\\&
  + \big( n_f(3-n_s)^2-2(4-n_f)^2 \big) \frac{s(t^2-4st+s^2)}{2u^4},
  \\
  \hat{C}^{-+-+}_3(s,t) &= \nonumber\\&
  + \big( 2(1-2n_s)+n_f \big)(4-n_f) \frac{3s^2(2t-s)}{2u^4} \nonumber\\&
  - \big( n_f(3-n_s)^2-2(4-n_f)^2 \big) \frac{3s(t^2-4st+s^2)}{4u^4}.
  \label{eq:3lmpmp}
\end{align}
It is easy to check that these coefficients correctly reproduce the known result in $\mathcal{N}=4$
super-symmetric Yang-Mills \cite{Bern:2005iz}. An expression valid for $\mathcal{N}>0$ super-symmetric
generators can be found by setting $n_f=\cN$ and $n_s=\cN-1$. We also note that complicated flavour
structures that only appear in the $-+-+$ coefficients vanish in the $\mathcal{N}=2$ theory,
\begin{align}
  &n_f(3-n_s)^2-2(4-n_f)^2 \to (\cN-2)(\cN-4)^2, \\
  & 2(1-2n_s)+n_f \to -3(\cN-2).
\end{align}
This shows that the integral basis for $\mathcal{N}=2$ is simpler than that of $\mathcal{N}=1$, a
feature which is new for four-dimensional maximal cuts at three-loops. At one-loop both $\mathcal{N}=1$ and $\mathcal{N}=2$ 
have the same integral basis with no rational terms. At two-loops the two theories
could differ in the lower propagator integrals but not in the maximal cut terms. We have also checked that by
taking two-particle cuts the full integrand factorizes onto the two-loop results for the full
integrand \cite{Badger:2012dp}.

\section{Conclusions}

In this paper we have shown that an integrand reduction technique based on computational algebraic
geometry techniques can be generalized to three-loop amplitudes. As a first step in this direction
we considered maximal cuts in four dimensions for the planar triple box topology contributing to
massless $2\to2$ scattering.

By using primary decomposition the 14 independent branches of the 10 on-shell conditions could be
found automatically. An explicit parametrisation of these solutions was found such that the
integrand would take the form of a simple Laurent expansion in two free variables. The integrand on
each solution factorises into a product of ten tree amplitudes which can be used to extract a set of
622 terms in the Laurent expansion.

Having derived a minimal parametrisation of the integrand in terms of 7 irreducible scalar
products we were able to construct an invertible matrix to map the coefficients of the Laurent
expansion to the coefficients of the ISPs. Though the matrix was quite large it was possible to invert
the system efficiently by using a branch-by-branch reconstruction using Gr\"obner bases and the
intersection of ideals.

Finally we were able to reduce the whole expression to master integrals using the {\tt Reduze2}
package for integration by parts identities.  This whole procedure derives a complete reduction for
an arbitrary process valid in any renormalizable gauge theory. As an application of the technique we
computed the MI coefficients for gluon-gluon scattering in Yang-Mills theory with adjoint scalars
and fermions.

Though a long way from a complete four-point amplitude in QCD the computation presented here can be
seen as a small step in the right direction. We hope the techniques presented here will be useful
when making the necessary generalizations to $D$-dimensional cuts and fewer particle cuts.

\acknowledgments{%
We would like to thank Andreas von Manteuffel and Cedric Studerus for assistance with the {\tt
Reduze2} package. We are also grateful to Pierpaolo Mastrolia and Bo Feng for useful discussions
and Alberto Guffanti for careful reading of the manuscript.
}

\appendix

\section{Branch-by-branch polynomial fitting method \label{app:branch}}

Consider generalized unitarity for an $L$-loop diagram. Let
$R$ be the polynomial ring of ISPs and $I$ be the ideal generated by
cut equations.

By the integrand-level reduction algorithm, the original integrand numerator is a polynomial $
\tilde N\in R$, while $[ \tilde N]$  is its image  in the quotient ring $R/I$.  $[\tilde N]$ can be
expanded over the {\it integrand basis}. The reduced integrand $N$, which is the simplest
representative for $[N]$, ($[N]=[\tilde N]$ in $R/I$), can be obtained by dividing $\tilde N$
towards the Gr\"obner basis $G(I)$ of $I$ \cite{Zhang:2012ce, Mastrolia:2012an} .

Often the cut solution have several branches. In other words, by
{\it primary decomposition}, $I$ is decomposed to the intersection
of several primary ideals \cite{Zhang:2012ce},
\begin{equation}
I=\bigcap_{i=1}^n I_i
\end{equation}
where $n$ is the number of branches.

Note that $I\subset I_i$, so $I_i$ contains more constraints than
$I$ itself. Hence the integrand can be reduced further on each
branch. We have the map,
\begin{eqnarray}
  \label{eq:1}
  p:  R/I &\to& R/I_1 \oplus \ldots \oplus R/I_n, \nonumber \\
   {[N]} &\mapsto& ([N]_1, \ldots, \mathcal [N]_n),
\end{eqnarray}
where $[\ldots ]$ and $[\ldots]_i$ stand for the equivalence relations
in $R/I$ and $R/I_i$.  If $N$ is known, it is straightforward
to get the simplest representative $N_i$ for $[N]_i$, by calculating the Gr\"obner basis for each
$I_i$. In general, $N_i$ contains much fewer terms than $N$.

So it is much easier to fit
each $N_i$ than $N$,  from the product of tree amplitudes. After all
$N_i$ are fitted,
the goal is to determine $N$ from $N_i$'s. The existence of $N$ is
guaranteed by the existence of the original numerator $\tilde N$. We
just need the following uniqueness condition:
\begin{lemma}
  The map $p$ is injective. In other words, if $N$ exists, then
  $N_1, \ldots , N_n$ uniquely determine
  $[N]$.
\end{lemma}
\begin{proof}
  Assume that there are two polynomials $N$ and
  $N'$ such that $p( N)=p( N')$. So $[N]_i=[N']_i=[N_i]_i$, $\forall i$.
  This means that $ N- N'\in I_i$, $\forall i$ and
  \begin{equation}
    \label{eq:2}
     N- N'\in \bigcap_{i=1}^n I_i=I,
  \end{equation}
so $[N]=[N'$] in $R/I$.
\end{proof}
This is an analogy of {\it Chinese remainder theorem}, however we do not
need the coprime condition since the existence of $N$ is guaranteed by
physics. After $[N]$ is determined, it is straightforward
to find its simplest representative $N$ by Gr\"obner basis $G(I)$ of
$I$.

In practice, we present the
following algorithm to get $N$ from $N_1,\ldots , N_n$. First, we consider the case $n=2$,
\begin{enumerate}
  \item  For two ideals $I_1=\langle f_1, \ldots, f_{m_1}\rangle$ and
  $I_2=\langle h_1 \ldots h_{m_2} \rangle$, Calculate the Gr\"obner basis $G(I_1+I_2)=\{g_1 , \ldots g_m\}$ for the ideal
    $I_1+I_2$. Record the
    transform matrix,
    \begin{equation}
      \label{eq:3}
      g_i =\sum_{j=1}^{m_1} a_{i j} f_j + \sum_{k=1}^{m_2} b_{i k} h_k.
    \end{equation}
\item Divide $N_1-N_2$ towards $G(I_1+I_2)$, where $r$ is the remainder,
  \begin{equation}
    \label{eq:4}
    N_1-N_2=\sum_{i=1}^{m} \psi_i g_i + r.
  \end{equation}
\item
If $r=0$, rewrite
   \begin{equation}
    \label{eq:4b}
    N_1-N_2=\sum_i \sum_j \psi_i a_{i j} f_j  + \sum_i \sum_k \psi_i
    b_{i k} h_k \equiv F_1 + F_2,
  \end{equation}
where $F_1 \in I_1$ and $F_2\in I_2$. Then $\hat N=N_1-F_1$. If $r\not=0$, print warning message and stop.
\item Calculate the intersection $I_1\cap I_2$ and its Gr\"obner basis
  $G(I_1\cap I_2)$. Divide $\hat N$ towards $G(I_1\cap I_2)$ and the
  remainder is $N$.
\end{enumerate}
The validity of this algorithm is self-evident since
$\hat N=N_1-F_1=N_2+F_2$ so $[\hat N]_1=[N_1]_1$ and $[\hat
N]_2=[N_2]_2$. And as long as $N$ exists, $r$ must be zero.

For cases with more than $2$ branches, we just need to repeat the
above algorithm for $n-1$ times,
\begin{enumerate}
\item Let $J=I_1$, $f=N_1$.
\item For $i=1$ to $n-1$
  \begin{itemize}
  \item Carry out the 2-branch algorithm for polynomials $(f,N_{i+1})$ and the two
    ideals $(J,I_{i+1})$. Then redefine $f$ as the output polynomial.
  \item $J=J\cap I_{i+1}$.
\end{itemize}
\item $N=f$.
\end{enumerate}
The validity can be checked by induction. We realise this algorithm in
a Macaulay2 program \cite{M2}.

\section{Flavour Configurations in the Triple Box \label{app:flav}}

The pre-factors of the 34 configurations appearing in \eqref{eq:treeprod} are
given in Table \ref{tab:flavs}. We note that by combining these configurations
together with different colour factors the results presented here would also be valid in QCD.
The symmetry factor of $-1$ for each fermion loop is included in the pre-factor.

For completeness we also present the well-known tree level amplitudes used in this paper,
\begin{align}
  -iA_3^{(0)}(1_g^-,2_g^-,3_g^+) &= \frac{\A12^3}{\A23\A31} \\
  -iA_3^{(0)}(1_\qb^-,2_q^+,3_g^-) &= \frac{\A13^2}{\A12} \\
  -iA_3^{(0)}(1_\sb,2_s,3_g^-) &= \frac{\A13\A23}{\A12} \\
  -iA_3^{(0)}(1_q^-,2_q^+,3_s) &= \A12
  \label{eq:3trees}
\end{align}

\begin{align}
  -iA_4^{(0)}(1_g^-,2_g^-,3_g^+,4_g^+) &= \frac{\A12^3}{\A23\A34\A41} \\
  -iA_4^{(0)}(1_g^-,2_g^+,3_g^+,4_g^-) &= \frac{\A41^3}{\A12\A23\A34} \\
  -iA_4^{(0)}(1_g^-,2_g^+,3_g^-,4_g^+) &= \frac{\A13^4}{\A12\A23\A34\A41}
  \label{eq:4trees}
\end{align}

\begin{table}
  \centering
  \begin{tabular}[h]{|ccccc|c||ccccc|c|}
    \hline
    $f_1$ & $f_2$ & $f_3$ & $f_4$ & $f_5$ & $F_{f_1\cdots f_5}$ &
    $f_1$ & $f_2$ & $f_3$ & $f_4$ & $f_5$ & $F_{f_1\cdots f_5}$ \\
    \hline
    $g$ &   $g$ &   $g$ &   $g$ &   $g$ &   $1$ &
    $g$ &   $g$ &   $q$ & $g$ &   $q$ & $-n_f$ \\
    $g$ &   $q$ & $g$ &   $q$ & $q$ &   $-n_f$ &
    $q$ &   $g$ &   $g$ &   $q$ &   $g$ &   $-n_f$ \\
    $g$ &   $q$ & $q$ & $q$ & $g$ &   $-n_f$ &
    $q$ &   $q$ & $g$ &   $g$ &   $q$ &   $-n_f$ \\
    $q$ &   $q$ & $q$ & $g$ &   $g$ &   $-n_f$ &
    $g$ &   $g$ &   $s$ & $g$ &   $s$ & $n_s$ \\
    $g$ &   $s$ & $g$ &   $s$ & $s$ &   $n_s$ &
    $s$ &   $g$ &   $g$ &   $s$ &   $g$ &   $n_s$ \\
    $g$ &   $s$ & $s$ & $s$ & $g$ &   $n_s$ &
    $s$ &   $s$ & $g$ &   $g$ &   $s$ &   $n_s$ \\
    $s$ &   $s$ & $s$ & $g$ &   $g$ &   $n_s$ &
    $g$ &   $s$ & $q$ & $s$ & $q$ &   $-n_fn_s$ \\
    $s$ &   $s$ & $q$ & $g$ &   $q$ &   $-n_fn_s$ &
    $g$ &   $q$ & $s$ & $q$ & $q$ & $-n_fn_s$ \\
    $s$ &   $q$ & $g$ &   $q$ &   $q$ &   $-n_fn_s$ &
    $q$ &   $s$ & $s$ & $q$ & $g$ &   $-n_fn_s$ \\
    $q$ &   $s$ & $g$ &   $q$ & $s$ &   $-n_fn_s$ &
    $g$ &   $q$ & $q$ &   $q$ & $s$ &   $-n_fn_s$ \\
    $s$ &   $q$ & $q$ & $q$ &   $g$ &   $-n_fn_s$ &
    $q$ &   $q$ & $s$ & $g$ &   $q$ & $-n_fn_s$ \\
    $q$ &   $q$ &   $g$ &   $s$ &   $q$ & $-n_fn_s$ &
    $q$ &   $q$ & $q$ &   $g$ &   $s$ &   $-n_fn_s$ \\
    $q$ &   $q$ &   $q$ &   $s$ &   $g$ &   $-n_fn_s$ &
    $s$ &   $g$ &   $q$ & $s$ &   $q$ & $-n_fn_s$ \\
    $q$ &   $g$ &   $s$ & $q$ &   $s$ & $-n_fn_s$ &
    $q$ &   $g$ &   $q$ & $q$ &   $q$ & $n_f^2$ \\
    $s$ &   $g$ &   $s$ & $s$ &   $s$ & $n_s^2$ &
    $s$ &   $q$ &   $s$ & $q$ & $q$ &   $-n_fn_s$ \\
    $s$ &   $q$ & $q$ &   $q$ &   $s$ &   $-n_fn_s$ &
    $q$ &   $q$ &   $s$ & $s$ & $q$ &   $-n_fn_s$ \\
    $q$ &   $q$ &   $q$ & $s$ & $s$ &   $-n_fn_s$ &
    $q$ &   $s$ & $q$ & $q$ & $q$ &   $2n_fn_s$ \\
    \hline
  \end{tabular}
  \caption{The definitions of the flavour pre-factors in eq. \eqref{eq:treeprod}.}
  \label{tab:flavs}
\end{table}

\providecommand{\href}[2]{#2}\begingroup\raggedright\endgroup

\end{document}